\documentclass{llncs}
\usepackage{amsmath,amssymb}
\usepackage[algoruled,linesnumbered,algo2e,vlined]{algorithm2e}
\usepackage{graphicx}
\usepackage{wrapfig}
\newcommand{\derive}{\mathit{val}}
\newcommand{\deriveInt}{\mathit{itv}}
\newcommand{\VarOcc}{\mathit{vOcc}}

\newcommand{\variable}[1]{X_{\langle #1\rangle}}
\newcommand{\kushi}{\xi_\mathcal{T}}
\newcommand{\suffix}{\mathit{suf}}
\newcommand{\prefix}{\mathit{pre}}

\newcommand{\Substr}{\mathit{Substr}}
\newcommand{\Suffix}{\mathit{Suffix}}
\newcommand{\SuffixTree}{\mathit{ST}}

\author{
  Hideo Bannai
  \and
  Shunsuke Inenaga
  \and
  Masayuki Takeda
}

\institute{
  Department of Informatics, Kyushu University\\
  \email{\{bannai,inenaga,takeda\}@inf.kyushu-u.ac.jp}\\
}
\title{
  Efficient LZ78 factorization of grammar compressed text
}
\date{}
\begin{document}
\maketitle
\begin{abstract}
  We present an efficient algorithm for computing the LZ78
  factorization of a text, where the text is represented
  as a straight line program (SLP), which is a context free grammar
  in the Chomsky normal form that generates a single string.
  Given an SLP of size $n$ representing a text $S$ of length $N$,
  our algorithm computes the LZ78 factorization of $T$
  in $O(n\sqrt{N}+m\log N)$ time and $O(n\sqrt{N}+m)$ space,
  where $m$ is the number of resulting LZ78 factors.
  We also show how to 
  improve the algorithm so that the $n\sqrt{N}$ term
  in the time and space complexities becomes either
  $nL$, where $L$ is the length of the longest LZ78 factor, or
  $(N - \alpha)$ where $\alpha \geq 0$
  is a quantity which depends on the amount of redundancy that 
  the SLP captures with respect to substrings of $S$ of a certain length.
  Since $m = O(N/\log_\sigma N)$ where 
  $\sigma$ is the alphabet size,
  the latter is asymptotically at least as fast as a linear
  time algorithm which runs on the uncompressed string when $\sigma$ is constant,
  and can be more efficient when
  the text is compressible, i.e. when $m$ and $n$ are small.
\end{abstract}
\section{Introduction}
Large scale textual data are usually stored in compressed form,
while it is later decompressed to be used.
In order to circumvent the computational resources required to 
handle and process the cumbersome uncompressed string, 
the compressed string processing (CSP) approach has been gaining attention.
The aim of CSP is to process text given in compressed form
without explicitly decompressing the entire text, 
therefore allowing space efficient, as well as
time efficient processing of the text when it is sufficiently compressed.

Many CSP algorithms work on a representation of the compressed text
called {\em straight line programs} (SLPs).
An SLP is a context free grammar in the Chomsky normal form that
derives a single string.
SLPs can efficiently model the outputs of many different types of
compression algorithms (e.g.: grammar
based~\cite{SEQUITUR,LarssonDCC99}, dictionary
based~\cite{LZ77,LZ78}), and hence, an algorithm that works on an SLP
can be applied to texts compressed by various compression
algorithms.
On the other hand, there are many CSP algorithms which make use of
specific properties that are implicit in the compressed representation
$C(S)$ of text $S$ obtained by using a certain compression algorithm
$C$~\cite{BilleFG09,freschi10:_lz78,gawrychowski11:_optim_lzw,gawrychowski12:_tying_lzw}.
Such CSP algorithms cannot be applied to representations
produced by any arbitrary compression algorithm.
To overcome this problem, we consider the problem of
computing the compressed representation $C(S)$ 
from an arbitrary SLP representing $S$,
without completely decompressing the SLP.

In this paper, we focus on the well known LZ78 compression
algorithm~\cite{LZ78}.
LZ78 compresses a given text based on a dynamic dictionary
which is constructed by partitioning the input string,
the process of which is called LZ78 factorization.
Other than its obvious use for compression, the LZ78 factorization
is an important concept used in various string processing
algorithms and
applications~\cite{crochemore03:_subquad_sequen_align_algor_unres_scorin_matric,li05:_lz78_based_strin_kernel,li05:_genre_class_lz78_based_strin_kernel,li06:_image_class_via_lz78_based_strin_kernel}.
The contribution of this paper is an 
$O(n\sqrt{N} + m\log N)$ time and $O(n\sqrt{N}+m)$ space algorithm 
to compute the LZ78 factorization of a string given as an SLP,
where $N$ is the length of the string, $n$ is the size of the SLP, and
$m$ is the number of LZ78 factors.

We further show how to improve the $n\sqrt{N}$ term in the time and
space complexities in two ways.
An application of doubling search enables the term to be reduced to
$nL$, where $L$ is the longest LZ78 factor.
Also, by applying the recent techniques of~\cite{goto12:_speed},
the term can be reduced to $N - \alpha$, where $\alpha \geq 0$
is a quantity which depends on the amount of redundancy that 
the SLP captures with respect to substrings of $S$ of a certain length.
Since it is known that $m = O(N/\log_\sigma N)$~\cite{LZ78},
where $\sigma$ is the alphabet size,
our approach is guaranteed to be asymptotically at least as fast as a linear
time algorithm which runs on the uncompressed string if $\sigma$ is considered constant,
and can be even more efficient when
the text is compressible, i.e. when $m$ and $n$ are small.

As a byproduct of the above results, 
we also obtain an efficient algorithm which converts 
a given \emph{LZ77 factorization} of a string~\cite{LZ77}
to the corresponding LZ78 factorization without explicit decompression.
We conclude the paper by mentioning several other interesting
potential applications of our algorithm.

\subsection*{Related Work}
An efficient algorithm for computing the LZ78 factorization was presented
in~\cite{jansson07:_compr_dynam_tries_applic_lz}.
Their algorithm requires only
$O(N(\log\sigma + \log\log_{\sigma} N)/ \log_{\sigma} N)$
bits of working space and
runs in $O(N(\log \log N)^2/(\log_{\sigma} N \log \log \log N))$
worst-case time
which is sub-linear when $\sigma = 2^{o(\log N \frac{\log \log \log N}{(\log \log N)^2})}$.
However, their input assumes the uncompressed text and it 
is unknown how to apply their algorithm without completely 
decompressing the SLP.

\section{Preliminaries}

\subsection{Strings}
Let $\Sigma$ be a finite {\em alphabet}
and $\sigma = |\Sigma|$.
An element of $\Sigma^*$ is called a {\em string}.
The length of a string $S$ is denoted by $|S|$. 
The empty string $\varepsilon$ is a string of length 0,
namely, $|\varepsilon| = 0$.
For a string $S = XYZ$, $X$, $Y$ and $Z$ are called
a \emph{prefix}, \emph{substring}, and \emph{suffix} of $S$, respectively.
The set of all substrings of a string $S$ is denoted by $\Substr(S)$.
The $i$-th character of a string $S$ is denoted by 
$S[i]$ for $1 \leq i \leq |S|$,
and the substring of a string $S$ that begins at position $i$ and
ends at position $j$ is denoted by $S[i:j]$ for $1 \leq i \leq j \leq |S|$.
For convenience, let $S[i:j] = \varepsilon$ if $j < i$.
For a string $S$ and integer $q \geq 0$,
let $\prefix(S,q)$ and $\suffix(S,q)$
represent respectively, the length-$q$ prefix and suffix of $T$,
that is,
$\prefix(S,q) = S[1:\min\{q,|S|\}]$ and $\suffix(S,q) = S[\max\{1,|S|-q+1\}:|S|]$.
We also assume that the last character of the string is a special
character `\$' that does not occur anywhere else in the string.

Our model of computation is the word RAM:
We shall assume that the computer word size is at least $\log |S|$, 
and hence, standard operations on
values representing lengths and positions of string $S$
can be manipulated in constant time.
Space complexities will be determined by the number of computer words (not bits).

\subsection{Straight Line Programs}
A {\em straight line program} ({\em SLP}) is a set of assignments 
$\mathcal T = \{ X_1 \rightarrow expr_1, X_2 \rightarrow expr_2, \ldots, X_n \rightarrow expr_n\}$,
where each $X_i$ is a distinct non-terminal variable and each $expr_i$ is an expression
that can be either $expr_i = a$ ($a\in\Sigma$), or $expr_i = X_{\ell(i)} X_{r(i)}$~($i
> \ell(i),r(i)$). An SLP is essentially a context free grammar in the
Chomsky normal form, that derives a single string.
Let $\derive(X_i)$ represent the string derived from variable $X_i$.
To ease notation, we sometimes associate $\derive(X_i)$ with $X_i$ and
denote $|\derive(X_i)|$ as $|X_i|$.
An SLP $\mathcal{T}$ {\em represents} the string $T = \derive(X_n)$.
The \emph{size} of the program $\mathcal T$ is the number $n$ of
assignments in $\mathcal T$. 

The derivation tree of SLP $\mathcal{T}$ is a labeled
ordered binary tree where each internal node is labeled with a
non-terminal variable in $\{X_1,\ldots,X_n\}$, and each leaf is labeled with a terminal character in $\Sigma$.
The root node has label $X_n$.
Let $\mathcal{V}$ denote the set of internal nodes in
the derivation tree.
For any internal node $v\in\mathcal{V}$, 
let $\langle v\rangle$ denote the index of its label
$\variable{v}$.
Node $v$ has a single child which is a leaf labeled with $c$
when $(\variable{v} \rightarrow c) \in \mathcal{T}$ for some $c\in\Sigma$,
or
$v$ has a left-child and right-child respectively denoted $\ell(v)$ and $r(v)$,
when
$(\variable{v}\rightarrow \variable{\ell(v)}\variable{r(v)}) \in \mathcal{T}$.
Each node $v$ of the tree derives $\derive(\variable{v})$,
a substring of $T$,
whose corresponding interval $\deriveInt(v) = [b:e]$,
with $T[b:e] = \derive(\variable{v})$,
can be defined recursively as follows.
If $v$ is the root node, then $\deriveInt(v) = [1:|T|]$.
Otherwise, if $(\variable{v}\rightarrow
\variable{\ell(v)}\variable{r(v)})\in\mathcal{T}$,
then,
$\deriveInt(\ell(v)) = [b_v:b_v+|\variable{\ell(v)}|-1]$
and
$\deriveInt(r(v)) = [b_v+|\variable{\ell(v)}|:e_v]$,
where $[b_v:e_v] = \deriveInt(v)$.
Let $\VarOcc(X_i)$ denote the number of times a variable $X_i$ occurs
in the derivation tree, i.e., 
$\VarOcc(X_i) = |\{ v \mid \variable{v}=X_i\}|$.

For any interval $[b:e]$ of $T (1\leq b < e \leq |T|)$, 
let $\kushi(b,e)$ denote the deepest node $v$ in the derivation tree,
which derives an interval containing $[b:e]$, that is,
$\deriveInt(v)\supseteq [b:e]$,
and no proper descendant of $v$
satisfies this condition.
We say that node $v$ {\em stabs} interval $[b:e]$,
and $\variable{v}$ is called the variable that stabs the interval.
We have $(\variable{v} \rightarrow
\variable{\ell(v)}\variable{r(v)})\in\mathcal{T}$,
$b\in \deriveInt(\ell(v))$, and  $e\in\deriveInt(r(v))$.
When it is not confusing, we will sometimes 
use $\kushi(b,e)$ to denote the variable $\variable{\kushi(b,e)}$.

SLPs can be efficiently pre-processed to hold various information.
$|X_i|$ and $\VarOcc(X_i)$ can be computed for all variables
$X_i~(1\leq i\leq n)$ in a total of $O(n)$ time by a simple dynamic
programming algorithm.

\begin{figure}
  \centerline{\includegraphics[width=0.6\textwidth]{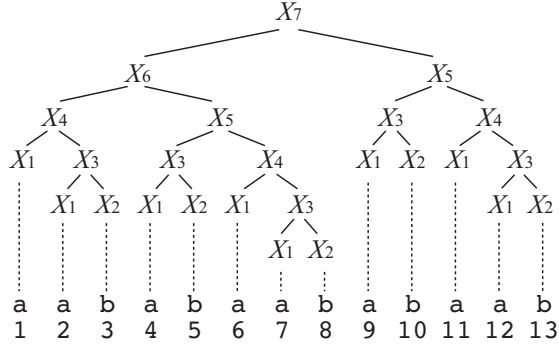}}
  \caption{
    The derivation tree of
    SLP $\mathcal T = \{ X_1 \rightarrow \mathtt{a}$, $X_2 \rightarrow \mathtt{b}$, $X_3 \rightarrow X_1X_2$,
    $X_4 \rightarrow X_1X_3$, $X_5 \rightarrow X_3X_4$, $X_6
    \rightarrow X_4X_5$, $X_7 \rightarrow X_6X_5 \}$.
    $T = \derive(X_7) = \mathtt{aababaababaab}$.
  }
  \label{fig:SLP}
\end{figure}

\subsection{LZ78 Encoding}
\label{subsec:lz78}
\begin{definition}[LZ78 factorization]
  The LZ78-factorization of a string $S$ is the factorization 
  $f_1 \cdots f_m$ of $S$, where each LZ78-factor
  $f_i\in\Sigma^+$ $(1\leq i \leq m)$ is the longest prefix of
  $f_i \cdots f_m$, such that
  $f_i \in \{f_jc \mid 1 \leq j < i, c\in\Sigma \}\cup\Sigma$.
\end{definition}

\begin{wrapfigure}[18]{r}{4cm}
  \vspace{-0.8cm}
  \centerline{
    \includegraphics[width=3cm]{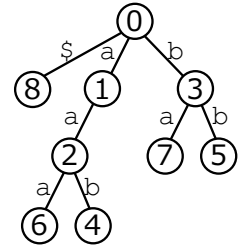}
  }
  \caption{The LZ78 dictionary for the string $\mathtt{aaabaabbbaaaaaaaba\$}$.
    Each node numbered $i$ represents the factor $f_i$ of the LZ78 factorization,
    where $f_i$ is the path label from the root to the node, 
    e.g.: $f_2 = \mathtt{aa}$, $f_4 = \mathtt{aab}$.
  }
  \label{fig:lz78dict}
\end{wrapfigure}

For a given string $S$, 
let $m$ denote the number of factors in its LZ78 factorization.
The LZ78 factorization of the string can be encoded by a sequence of
pairs, where the pair for factor $f_i$ consists of the ID $j$ of the
previous factor $f_j$ ($j = 0$ and $f_0 = \varepsilon$ when there is none) 
and the new character $S[|f_1\cdots f_{i}|]$.
Regarding this pair as a parent and edge label, the factors can also
be represented as a trie. (See Fig.~\ref{fig:lz78dict}.)

By using this trie, the LZ78 factorization of a string of length $N$
can be easily computed incrementally in $O(N\log\sigma)$ time and
$O(m)$ space;
Start from an empty tree with only the root.
For $1 \leq i\leq m$,
to calculate $f_i$, let $v$ be the node of the trie reached by traversing
the tree with $S[p:q]$, where $p = |f_0\cdots f_{i-1}|+1$, 
and $q\geq p$ is the smallest position after $p$
such that $v$ does not have an outgoing edge labeled with $S[q+1]$.
Naturally, $v$ represents the longest previously used LZ78-factor
that is a prefix of $S[p:|S|]$. 
Then, we can insert an edge labeled with $S[q+1]$ to a new node
representing factor $f_i$, branching from $v$.
The update for each factor $f_i$ can be done in $O(|f_i|\log\sigma)$ time
for the traversal and in $O(\log\sigma)$ time for the insertion,
with a total of $O(N\log\sigma)$ time for all the factors.
Since each node of the trie except the root corresponds to an LZ78 factor,
the size of the trie is $O(m)$.

\begin{example}
  The LZ78 factorization of string
  $\mathtt{aaabaabbbaaaaaaaba\$}$ is
  $\mathtt{a}$,
  $\mathtt{aa}$,
  $\mathtt{b}$,
  $\mathtt{aab}$,
  $\mathtt{bb}$,
  $\mathtt{aaa}$,
  $\mathtt{aaaa}$,
  $\mathtt{ba}$,
  $\mathtt{\$}$, and can be represented as
  $(0,\mathtt{a})$,
  $(1,\mathtt{a})$,
  $(0,\mathtt{b})$,
  $(2,\mathtt{b})$,
  $(3,\mathtt{b})$,
  $(2,\mathtt{a})$,
  $(6,\mathtt{a})$,
  $(3,\mathtt{a})$,
  $(0,\mathtt{\$})$.
\end{example}

\subsection{Suffix Trees}
We give the definition of a very important and well known string index
structure, the suffix tree.
To assure property~\ref{def:suffixtreeleaf} for the sake of presentation,
we assume that the string ends with a unique symbol
that does not occur elsewhere in the string.
\begin{definition}[Suffix Trees~\cite{Weiner}]
  For any string $S$, its suffix tree, denoted 
  $\SuffixTree(S)$, is a labeled rooted tree which satisfies the following:
  \begin{enumerate}
   \item each edge is labeled with an element in $\Sigma^+$;
   \item there exist exactly $n$ leaves, where $n = |S|$;
   \item\label{def:suffixtreeleaf} for each string $s \in \Suffix(S)$, there is a unique path from the root 
     to a leaf which spells out $s$;
   \item each internal node has at least two children; 
   \item
     the labels $x$ and $y$ of any two distinct out-going edges from the same node begin
     with different symbols in $\Sigma$
  \end{enumerate}
\end{definition}

Since any substring of $S$ is a prefix of some suffix of $S$,
positions in the suffix tree of $S$ correspond to a substring of $S$
that is represented by the string spelled out on the path from the
root to the position.
We can also define a {\em generalized} suffix tree of a set of strings,
which is simply the suffix tree that contains all
suffixes of all the strings in the set.

It is well known that suffix trees can be represented and constructed
in linear time~\cite{Weiner,McC76,Ukk95}, even independently of the
alphabet size for integer alphabets~\cite{farach97:_optim_suffix_tree_const_large_alphab}.
Generalized suffix trees for a set of strings
$\mathbf{S} =\{S_1,\ldots, S_k\}$, can be constructed in linear time in the
total length of the strings, by simply constructing the suffix tree of the
string $S_1\$_1\cdots S_k\$_k$, and pruning the tree below the first
occurrence of any $\$_i$, where $\$_i$ $(1\leq i\leq k)$
are unique characters that do not occur elsewhere in strings of $\mathbf{S}$.

\section{Algorithm}
We describe our algorithm for computing the LZ78 factorization 
of a string given as an SLP in two steps.
The basic structure of the algorithm follows the simple LZ78
factorization algorithm for uncompressed strings that uses a trie
as mentioned in Section~\ref{subsec:lz78}.
Although the space complexity of the trie is only $O(m)$, we need some way to
accelerate the traversal of the trie
in order to achieve the desired time bounds.

\subsection{Partial Decompression}
We use the following property of LZ78 factors which is straightforward
from its definition.
\begin{lemma}\label{lemma:factor_lb_ub}
  For any string $S$ of length $N$ and its LZ78-factorization
  $f_1\cdots f_m$,
  $m \geq c_N$ and
  $|f_i| \leq c_N$ for all $1\leq i\leq m$, 
  where $c_N = \sqrt{2N+1/4}-1/2$.
\end{lemma}
\begin{proof}
  Since a factor can be at most 1 character longer than a previously
  used factor, $|f_i| \leq i$.
  Therefore, $N = \sum_{i=1}^m |f_i| \leq \sum_{i=1}^m i$,
  and thus $m \geq \sqrt{2N+1/4} - 1/2$.
  For any factor of length $x = |f_{i_x}|$, there exist
  distinct factors $f_{i_1},\ldots, f_{i_{x-1}}$ 
  whose lengths are respectively $1, \ldots, x-1$.
  Therefore, 
  $N = \sum_{i=1}^m|f_i| \geq \sum_{i=1}^x i$,
  and $x \leq \sqrt{2N+1/4}-1/2$.\qed
\end{proof}

The lemma states that the length of an LZ78-factor is bounded by
$c_N$.
To utilize this property, we use ideas similar to those developed
in~\cite{goto11:_fast_minin_slp_compr_strin,goto12:_speed} for
counting the frequencies of all substrings of a certain length
in a string represented by an SLP;
For simplicity, assume $c_N\geq 2$.
For each variable $X_i \rightarrow X_{\ell(i)}X_{r(i)}$,
any length $c_N$ substring that is stabbed by $X_i$ is a substring of
$t_i = \suffix(\derive(X_{\ell(i)}),c_N-1)\prefix(\derive(X_{r(i)}),c_N-1)$.
On the other hand, all length $c_N$ substrings are stabbed by
some variable. This means that if we consider the set of strings
consisting of $t_i$ for all variables such that
$|X_i|\geq c_N$,  any length $c_N$ substring of $S$
is a substring of at least one of the strings.
We can compute all such strings
$T_S = \{ t_i \mid |X_i|\geq c_N\}$ where
$(X_i \rightarrow X_{\ell(i)}X_{r(i)}) \in \mathcal{T}$
in time linear in the total length,
i.e. $O(nc_N)$ time by a straightforward dynamic programming~\cite{goto11:_fast_minin_slp_compr_strin}.

All length $c_N$ substrings of $S$ occur as
substrings of strings in $T_S$, and 
by Lemma~\ref{lemma:factor_lb_ub}, it follows that $T_S$ contains all
LZ78-factors of $S$ as substrings. 

\subsection{Finding the Next Factor}
In the previous subsection, we described how to partially
decompress a given SLP of size $n$ representing a string $S$ of length $N$,
to obtain a set of strings $T_S$ with total length $O(n\sqrt{N})$, 
such that any LZ78-factor of $S$ is a substring of at
least one of the strings in $T_S$.
We next describe how to identify these substrings.

We make the following key observation: since the LZ78-trie of a string $S$
is a trie composed by substrings of $S$, it can be superimposed
on a suffix tree of $S$, and be completely contained in it,
with the exception that some nodes of the trie may correspond to
implicit nodes of the suffix tree (in the middle of an edge of the
suffix tree).
Furthermore, this superimposition can also be done to the generalized suffix tree
constructed for $T_S$. (See Fig.~\ref{fig:superimpose}.)

\begin{figure}[t]
  \centerline{
    \includegraphics[width=0.7\textwidth]{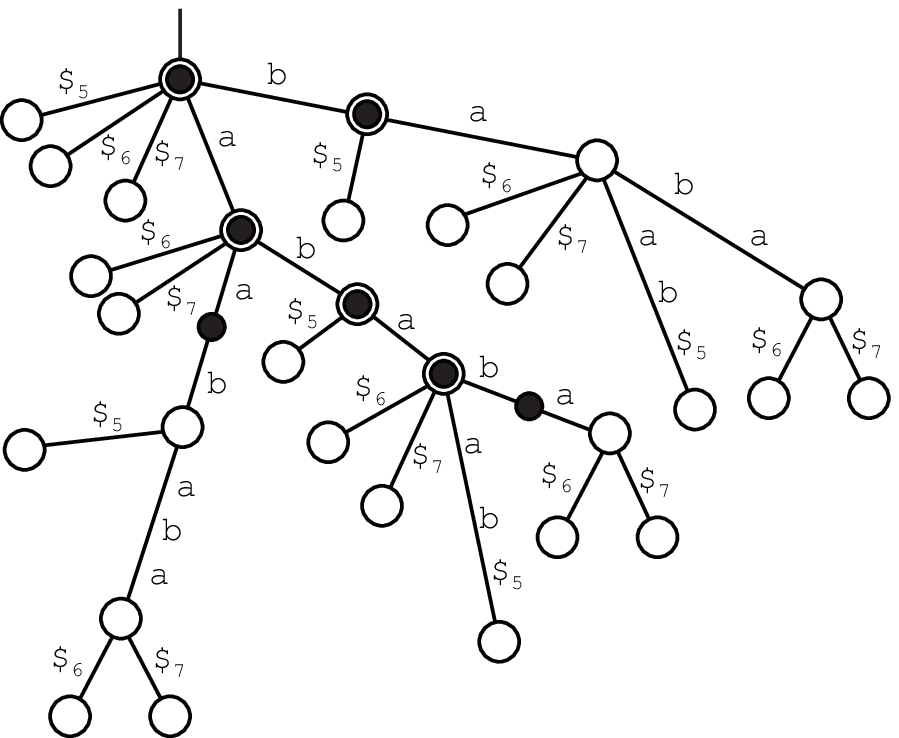}
  }
  \caption{
    The LZ78-trie of string $S = \mathtt{aababaababaab}$,
    superimposed on the generalized suffix tree of $T_S = \{ t_5,t_6,t_7 \} =
    \{\mathtt{abaab\$_5},\mathtt{aababa\$_6},\mathtt{aababa\$_7} \}$
    for the SLP of Fig.~\ref{fig:SLP}.
    Here, $\mathtt{\$_5,\$_6,\$_7}$ are end markers of each 
    string in $T_S$, introduced so that each position in a 
    string of $T_s$ corresponds to a leaf of the suffix tree.
    The subtree consisting of the dark nodes is the LZ78-trie,
    derived from the LZ78-factorization:
    $\mathtt{a},\mathtt{ab},\mathtt{aba},\mathtt{abab},\mathtt{aa},\mathtt{b}$,
    of $S$.
    Since any length $\lfloor c_N\rfloor = 4$ substring of $S$ is a
    substring of at least one string in $T_S$,
    any LZ78-factor of $S$ is a substring of some string of $T_S$, and
    the generalized suffix tree of $T_S$ completely includes the LZ78-trie.
  }
  \label{fig:superimpose}
\end{figure}
Suppose we have computed the LZ78 factorization 
$f_1\cdots f_{i-1}$, up to position $p - 1 = |f_1\cdots f_{i-1}|$,
and wish to calculate the next LZ78-factor starting at position $p$.
Let $v = \kushi(p,p+c_N-1)$,
let $X_j = \variable{v}$ be the variable that stabs the interval $[p:p+c_N-1]$,
let $q$ be the offset of $p$ in $t_j$,
and let $w$ be the leaf of the generalized suffix tree
that corresponds to the suffix $t_j[q:|t_j|]$. 
The longest previously used factor that is a prefix of $S[p:|S|]$
is the longest common prefix between $t_j[q:|t_j|]$
and all possible paths on the LZ78-trie built so far.
If we consider the suffix tree as a semi-dynamic tree, 
where nodes corresponding to the superimposed LZ78-trie are
dynamically added and marked, the node $x$ we seek
is the {\em nearest marked ancestor} of $w$.

The generalized suffix tree for $T_S$ can be computed in $O(n\sqrt{N})$ time.
We next describe how to obtain the values $v$, 
$q$ (and therefore $w$), and $x$
as well as the computational complexities involved.

A na\"ive algorithm for obtaining $v$ and $q$ would be to traverse
down the derivation tree of the SLP from the root, checking the 
decompressed lengths of the left and right child of each variable to determine which child
to go down, in order to find the variables that correspond to
positions $p$ and $p+c_N-1$. By doing the search in parallel,
we can find $v$ as the node at which the search for each position
diverges, i.e., the
lowest common ancestor of leaves in the derivation tree corresponding
to positions $p$ and $p+c_N-1$.
This traversal requires $O(h)$ time, where $h$ is the height of the
SLP, which can be as large as $O(n)$. 
To do this more efficiently, we can apply the algorithm
of~\cite{philip11:_random_acces_gramm_compr_strin}, which allows
random access to arbitrary positions of the SLP in $O(\log N)$ time,
with $O(n)$ time and space of preprocessing.

\begin{theorem}[\cite{philip11:_random_acces_gramm_compr_strin}]
  \label{thm:slprandomaccess}
  For an SLP of size $n$ representing a string of length $N$,
  random access can be supported in time $O(\log N)$ after 
  $O(n)$ preprocessing time and space in the RAM model.
\end{theorem}

Their algorithm basically constructs data structures in order
to simulate the traversal of the SLP from the root, but reduces the time complexity
from $O(h)$ to $O(\log N)$.
Therefore, by running two random access operations for positions $p$ and
$p+c_N-1$ in parallel until they first diverge, we can obtain $v$ in
$O(\log N)$ time.
We note that this technique is the same as the first part of
their algorithm for decompressing a substring
$S[i:j]$ of length $m = j-i+1$ in $O(m + \log N)$ time.
The offset of $p$ from the beginning of $\variable{v}$ 
can be obtained as a byproduct of the search
for position $p$, and therefore, $q$ can also be computed in $O(\log N)$ time.

For obtaining $x$, we use a data structure that maintains a rooted
dynamic tree with marked/un\-marked nodes such that the nearest marked
ancestor in
the path from a given node to the root can be found very efficiently.
The following result allows us to find $x$ -- the nearest marked
ancestor of $w$ -- in amortized constant time.
\begin{lemma}[\cite{westbrook92:_fast_increm_planar_testin,amir_improved_dynamic_1995}] \label{lem:nma}
A semi-dynamic rooted tree can be maintained in linear space
so that the following operations are supported in amortized $O(1)$ time:
1) find the nearest marked ancestor of any node;
2) insert an unmarked node;
3)  mark an unmarked node.
\end{lemma}

For inserting the new node for the new LZ78-factor,
we simply move down the edge of the suffix tree if $x$
was an implicit node and has only one child.
When $x$ is branching, we can
move down the correct suffix tree using level ancestor queries 
of the leaf $w$, therefore not requiring an $O(\log\sigma)$
factor.
\begin{lemma}[Level ancestor query~\cite{berkman94:_findin,bender04:_level_ances_probl}] \label{lem:laq}
Given a static rooted tree,
we can preprocess the tree in linear time and space 
so that the $\ell$th node in the path
from any node to the root can be found in $O(1)$ time
for any integer $\ell \geq 0$, if such exists.
\end{lemma}
Technically, our suffix tree is semi-dynamic in that new
nodes are created since the LZ78-trie is superimposed. 
However, since we are only interested in level ancestor queries at
branching nodes, we only need to answer them for the original suffix
tree. Therefore, we can preprocess the tree in $O(n\sqrt{N})$ time and
space to answer the level ancestor queries in $O(1)$ time.

The main result of this section follows:
\begin{theorem} \label{theo:SLP_2_LZ78}
 Given an SLP of size $n$ representing a string $S$ of length $N$,
 we can compute the LZ78 factorization of $S$ 
 in $O(n\sqrt{N} + m\log N)$ time and $O(n\sqrt{N} + m)$ space,
 where $m$ is the size of the LZ78 factorization.
\end{theorem}
A better bound can be obtained by employing a simple doubling search
on the length of partial decompressions.
\begin{corollary}
  \label{cor:slp2lz78}
 Given an SLP of size $n$ representing a string $S$ of length $N$,
 we can compute the LZ78 factorization of $S$ 
 in $O(nL + m\log N)$ time and $O(nL + m)$ space,
 where $m$ is the size of the LZ78 factorization, and $L$ is the
 length of the longest LZ78 factor.
\end{corollary}
\begin{proof}
Instead of using $c_N$ for the length of partial decompressions,
we start from length $2$. 
For some length $2^{i-1}$, if the LZ78 trie {\em outgrows} the suffix
tree and reaches a leaf, we rebuild the suffix tree and 
the embedded LZ78 trie for length $2^i$ and continue with the
factorization. This takes $O(n2^i)$ time,
and the total asymptotic complexity becomes
$n (2+\cdots+2^{\lceil\log_2{L} \rceil}) = O(nL)$.
Notice that the $m\log N$ term does not increase, 
since the factorization itself is not restarted,
and also since the data structure
of~\cite{philip11:_random_acces_gramm_compr_strin}
is reused and only constructed once.
\qed
\end{proof}

\subsection{Reducing Partial Decompression}
\label{sec:reduce_decompressions}
By using the same techniques of~\cite{goto12:_speed},
we can reduce the partial decompression conducted on the SLP,
and reduce the complexities of our algorithm.
Let $I = \{i \mid |X_i| \geq c_N\} \subseteq [1:n]$.
The technique exploits the overlapping portions of each of the strings in $T_S$.
The algorithm of~\cite{goto12:_speed} shows how to construct, in time linear of its size,
a trie of size
$(c_N - 1) + \sum_{i \in I}(|t_i| - (c_N - 1))  = N - \alpha = N_\alpha$
such that there is a one to one correspondence between
a length $c_N$ path on the trie and
a length $c_N$ substring of a string in $T_S$.
Here,
\begin{equation}
 \alpha = \sum_{i \in I} ((\VarOcc(X_i) - 1) \cdot (|t_i| - (c_N -
 1))) \geq 0 \label{eqn:alpha}
\end{equation}
can be seen as a
quantity which depends on the amount of redundancy that the SLP
captures with respect to length $c_N$ substrings.

Furthermore, a suffix tree of a trie can be constructed in linear time:
\begin{lemma}[\cite{shibuya03:_const_suffix_tree_tree_large_alphab}]
  Given a trie, the suffix tree for the trie can be constructed in
  linear time and space.
\end{lemma}
The generalized suffix tree for $T_S$ used in our algorithm
can be replaced with the suffix
tree of the trie, and we can reduce the $O(n\sqrt{N})$ term in
the complexity to $O(N_\alpha)$, thus obtaining
an $O(N_\alpha+m\log N)$ time and $O(N_\alpha+m)$ space
algorithm.
Since $N_{\alpha}$ is also bounded by $O(n \sqrt{N})$, 
we obtain the following result:
\begin{theorem} \label{theo:faster_algo}
 Given an SLP of size $n$ representing a string $S$ of length $N$,
 we can compute the LZ78 factorization of $S$ 
 in $O(N_\alpha + m\log N)$ time and $O(N_\alpha + m)$ space,
 where $m$ is the size of the LZ78 factorization,
 $N_\alpha = O(\min\{N-\alpha, n \sqrt{N}\})$,
 and $\alpha\geq 0$ is defined as in Equation~(\ref{eqn:alpha}).
\end{theorem}

Since $m = O(N/\log_\sigma N)$~\cite{LZ78},
our algorithms are asymptotically at least as fast as a linear
time algorithm which runs on the uncompressed string when the
alphabet size is constant.
On the other hand, $N_\alpha$ can be much smaller than 
$O(n\sqrt{N})$ when $\VarOcc(X_i) > 1$  for many of the variables.
Thus our algorithms can be faster 
when the text is compressible, i.e., $n$ and $m$ are small.

\subsection{Conversion from LZ77 Factorization to LZ78 Factorization}

As a byproduct of the algorithm proposed above,
we obtain an efficient algorithm that converts 
a given \emph{LZ77 factorization}~\cite{LZ77} of a string to the corresponding LZ78 factorization,
without explicit decompression.

\begin{definition}[LZ77 factorization]
 The LZ77-factorization of a string $S$ is
 the factorization $f_1,\ldots,f_r$ of $S$ such that for every $i=1,\ldots,r$,
 factor $f_i$ is the longest prefix of $f_i \cdots f_r$ with $f_i \in F_i$,
 where
 $F_i=\Substr(f_1\cdots f_{i-1})\cup\Sigma$.
\end{definition}

It is known that the LZ77-factorization of 
string $S$ can be efficiently transformed into an SLP representing $S$.
\begin{theorem}[\cite{Rytter03}] \label{theo:LZ77_2_SLP}
Given the LZ77 factorization 
of size $r$ for a string $S$ of length $N$,
we can compute in $O(r \log N)$ time an SLP representing $S$,
of size $O(r \log N)$ and of height $O(\log N)$.
\end{theorem}

The following theorem is immediate from Corollary~\ref{cor:slp2lz78} 
and Theorem~\ref{theo:LZ77_2_SLP}.

\begin{theorem}
Given the LZ77 factorization of size $r$ for a string $S$ of length $N$,
we can compute the LZ78 factorization for $S$ in 
$O(r L \log N  + m \log N)$ time and $O(r L \log N + m)$ space,
where $m$ is the size of the LZ78 factorization for $S$,
and $L$ is the length of the longest LZ78 factor.
\end{theorem}

It is also possible to improve the complexities of the above theorem
using Theorem~\ref{theo:faster_algo},
so that the conversion from LZ77 to LZ78 can be conducted in
$O(N_\alpha + m\log N)$ time and $O(N_\alpha + m)$ space,
where $N_\alpha$ here is defined for the SLP generated from the input LZ77 factorization.
This is significant since the resulting algorithm is at least as efficient as 
a na\"ive approach which requires decompression of the input LZ77 factorization,
and can be faster when the string is compressible.

\section{Discussion}
We showed an efficient algorithm for calculating the 
LZ78 factorization of a string $S$, from an arbitrary SLP 
of size $n$ which represents $S$.
The algorithm is guaranteed to be
asymptotically at least as fast as a linear
time algorithm that runs on the uncompressed text, and can be much
faster when $n$ and $m$ are small, i.e., the text is compressible.

It is easy to construct an SLP of size $O(m)$ that represents
string $S$, given its LZ78 factorization whose size is $m$~\cite{KidaCollageTCS}.
Thus, although it was not our primary focus in this paper,
the algorithms we have developed can be regarded as a
{\em re-compression} by LZ78, of strings represented as SLPs.
The concept of re-compression was recently used to speed up
fully compressed pattern matching~\cite{jez12:_faster}.
We mention two other interesting potential applications
of re-compression, for which our algorithm provides solutions:

\subsection*{Maintaining Dynamic SLP Compressed Texts}
    Modification to the SLP corresponding to edit
    operations on the string that it represents, 
    e.g.: character substitutions, insertions, deletions
    can be conducted in $O(h)$ time,
    where $h$ is the height of the SLP.
    However, these modifications are
    {\em ad-hoc}, and there are no guarantees as to how
    compressed the resulting SLP is,
    and repeated edit operations will inevitably cause degradation on the
    compression ratio.
    By periodically re-compressing the SLP, we can maintain the compressed
    size (w.r.t. LZ78) of the representation, without having to explicitly
    decompress the entire string during the maintenance process.

\subsection*{Computing the NCD w.r.t. LZ78 without explicit decompression}
  The Normalized Compression Distance
  (NCD)~\cite{cilibrasi05:_clust_compr} 
  measures the distance between two data strings, based on a specific
  compression algorithm.
  It has been shown to be effective for various clustering and
  classification tasks, while not requiring in-depth prior knowledge of the data.
  NCD between two strings $S$ and $T$ w.r.t. compression algorithm $A$ is
  determined by the values $C_A(ST)$, $C_A(S)$, and $C_A(T)$,
  which respectively denote the sizes of the compressed representation of
  strings $ST$, $S$, and $T$ when compressed by algorithm $A$.
  
  When $S$ and $T$ are represented as SLPs, we can compute
  $C_{\mathrm{LZ78}}(S)$ and $C_{\mathrm{LZ78}}(T)$ without
  explicitly decompressing all of $S$ and $T$, using the algorithms in this paper.
  Furthermore, the SLP for the concatenation $ST$ can
  be obtained by simply considering a new single variable and production rule
  $X_{ST} \rightarrow X_{S}X_{T}$, where $X_S$ and $X_T$ are respectively
  the roots of the SLP which derive $S$ and $T$.
  Thus, by applying our algorithm on this SLP, we can compute
  $C_{\mathrm{LZ78}}(ST)$ without explicit decompression as well.
  Therefore it is possible to compute $NCD$ w.r.t. LZ78 between strings represented
  as SLPs, and therefore even cluster or classify them,
  without explicit decompression.

\section*{Acknowledgements}
We thank the anonymous reviewers for helpful comments to
improve the paper.

\bibliographystyle{splncs03}
\bibliography{ref}

\begin{thebibliography}{10}
\providecommand{\url}[1]{\texttt{#1}}
\providecommand{\urlprefix}{URL }

\bibitem{amir_improved_dynamic_1995}
Amir, A., Farach, M., Idury, R.M., Poutr{\'e}, J.A.L., Sch{\"a}ffer, A.A.:
  Improved dynamic dictionary matching. Information and Computation  119(2),
  258--282 (1995)

\bibitem{bender04:_level_ances_probl}
Bender, M.A., Farach-Colton, M.: The level ancestor problem simplified. Theor.
  Comput. Sci.  321(1),  5--12 (2004)

\bibitem{berkman94:_findin}
Berkman, O., Vishkin, U.: Finding level-ancestors in trees. J. Comput. System
  Sci.  48(2),  214--230 (1994)

\bibitem{BilleFG09}
Bille, P., Fagerberg, R., G{\o}rtz, I.L.: Improved approximate string matching
  and regular expression matching on {Ziv-Lempel} compressed texts. ACM
  Transactions on Algorithms  6(1) (2009)

\bibitem{philip11:_random_acces_gramm_compr_strin}
Bille, P., Landau, G.M., Raman, R., Sadakane, K., Satti, S.R., Weimann, O.:
  Random access to grammar-compressed strings. In: Proc. SODA 2011. pp.
  373--389 (2011)

\bibitem{cilibrasi05:_clust_compr}
Cilibrasi, R., Vit\'anyi, P.M.: Clustering by compression. IEEE Transactions on
  Information Theory  51(4),  1523--1545 (2005)

\bibitem{crochemore03:_subquad_sequen_align_algor_unres_scorin_matric}
Crochemore, M., Landau, G.M., Ziv-Ukelson, M.: A subquadratic sequence
  alignment algorithm for unrestricted scoring matrices. SIAM J. Comput.
  32(6),  1654--1673 (2003)

\bibitem{farach97:_optim_suffix_tree_const_large_alphab}
Farach, M.: Optimal suffix tree construction with large alphabets. In: Proc.
  FOCS 1997. pp. 137--143 (1997)

\bibitem{freschi10:_lz78}
Freschi, V., Bogliolo, A.: A faster algorithm for the computation of string
  convolutions using {LZ78} parsing. Information Processing Letters
  110(14--15),  609--613 (2010)

\bibitem{gawrychowski11:_optim_lzw}
Gawrychowski, P.: Optimal pattern matching in {LZW} compressed strings. In:
  Proc. SODA 2011. pp. 362--372 (2011)

\bibitem{gawrychowski12:_tying_lzw}
Gawrychowski, P.: Tying up the loose ends in fully {LZW}-compressed pattern
  matching. In: Proc. STACS 2012. pp. 624--635 (2012)

\bibitem{goto11:_fast_minin_slp_compr_strin}
Goto, K., Bannai, H., Inenaga, S., Takeda, M.: Fast $q$-gram mining on {SLP}
  compressed strings. In: Proc. SPIRE 2011. pp. 289--289 (2011)

\bibitem{goto12:_speed}
Goto, K., Bannai, H., Inenaga, S., Takeda, M.: Speeding up $q$-gram mining on
  grammar-based compressed texts. In: Proc. CPM 2012. pp. 220--231 (2012)

\bibitem{jansson07:_compr_dynam_tries_applic_lz}
Jansson, J., Sadakane, K., Sung, W.K.: Compressed dynamic tries with
  applications to {LZ}-compression in sublinear time and space. In: Proc.
  FSTTCS 2007. pp. 424--435 (2007)

\bibitem{jez12:_faster}
Je\.z, A.: Faster fully compressed pattern matching by recompression. In: Proc.
  ICALP 2012 (2012), (preprint: arXiv:1111.3244v2)

\bibitem{KidaCollageTCS}
Kida, T., Shibata, Y., Takeda, M., Shinohara, A., Arikawa, S.: Collage system:
  A unifying framework for compressed pattern matching. Theor. Comput. Sci.
  298(1),  253--272 (2003)

\bibitem{LarssonDCC99}
Larsson, N.J., Moffat, A.: Offline dictionary-based compression. In: Proc. DCC
  1999. pp. 296--305. IEEE Computer Society (1999)

\bibitem{li05:_genre_class_lz78_based_strin_kernel}
Li, M., Sleep, R.: Genre classification via an {LZ78}-based string kernel. In:
  Proc. ISMIR 2005. pp. 252--259 (2005)

\bibitem{li05:_lz78_based_strin_kernel}
Li, M., Sleep, R.: An {LZ78} based string kernel. In: Proc. ADMA 2005. pp.
  678--689 (2005)

\bibitem{li06:_image_class_via_lz78_based_strin_kernel}
Li, M., Zhu, Y.: Image classification via {LZ78} based string kernel: A
  comparative study. In: Proc. PAKDD 2006. pp. 704--712 (2006)

\bibitem{McC76}
McCreight, E.M.: A space-economical suffix tree construction algorithm. Journal
  of ACM  23(2),  262--272 (1976)

\bibitem{SEQUITUR}
Nevill-Manning, C.G., Witten, I.H., Maulsby, D.L.: Compression by induction of
  hierarchical grammars. In: Proc. DCC 1994. pp. 244--253 (1994)

\bibitem{Rytter03}
Rytter, W.: Application of {L}empel-{Z}iv factorization to the approximation of
  grammar-based compression. Theor. Comput. Sci.  302(1-3),  211--222 (2003)

\bibitem{shibuya03:_const_suffix_tree_tree_large_alphab}
Shibuya, T.: Constructing the suffix tree of a tree with a large alphabet.
  IEICE Transactions on Fundamentals of Electronics, Communications and
  Computer Sciences  E86-A(5),  1061--1066 (2003)

\bibitem{Ukk95}
Ukkonen, E.: On-line construction of suffix trees. Algorithmica  14(3),
  249--260 (1995)

\bibitem{Weiner}
Weiner, P.: Linear pattern-matching algorithms. In: Proc. of 14th IEEE Ann.
  Symp. on Switching and Automata Theory. pp. 1--11. Institute of Electrical
  Electronics Engineers, New York (1973)

\bibitem{westbrook92:_fast_increm_planar_testin}
Westbrook, J.: Fast incremental planarity testing. In: Proc. ICALP 1992. pp.
  342--353 (1992)

\bibitem{LZ77}
Ziv, J., Lempel, A.: A universal algorithm for sequential data compression.
  IEEE Transactions on Information Theory  IT-23(3),  337--349 (1977)

\bibitem{LZ78}
Ziv, J., Lempel, A.: Compression of individual sequences via variable-length
  coding. IEEE Transactions on Information Theory  24(5),  530--536 (1978)

\end{thebibliography}
\end{document}